\documentclass[10pt]{article}

\usepackage{amsmath, amsthm, amssymb}
\usepackage{fullpage}
\usepackage{verbatim}
\usepackage{enumerate}

\usepackage{tikz}
\usetikzlibrary{positioning}
\tikzstyle{every node}=[fill, circle]
\tikzstyle{every edge}=[draw]

\usepackage{makecell}

\newcommand{\tr}{\mathrm{Tr}}
\newcommand{\ee}{\mathbf{e}}

\newtheorem{theorem}{Theorem}[section]
\newtheorem{lemma}[theorem]{Lemma}
\newtheorem{prop}[theorem]{Proposition}
\newtheorem{cor}[theorem]{Corollary}

\theoremstyle{definition}
\newtheorem{example}[theorem]{Example}

\begin{document}

\title{Pretty Good State Transfer and Minimal Polynomials}
\author{Christopher M. van Bommel\footnote{Supported by a Pacific Institute for the Mathematical Sciences Post-doctoral Fellowship.}\\Department of Mathematics\\University of Manitoba\\Winnipeg, MB, Canada\\\texttt{Christopher.VanBommel@umanitoba.ca}}
\date{\today}

\maketitle

\begin{abstract}
We examine conditions for a pair of strongly cospectral vertices to have pretty good quantum state transfer in terms of minimal polynomials, and provide cases where pretty good state transfer can be ruled out.  We also provide new examples of simple, unweighted graphs exhibiting pretty good state transfer.  Finally, we consider modifying paths by adding symmetric weighted edges, and apply these results to this case.
\end{abstract}

\section{Introduction}

In quantum information processing, a key requirement is the ability to transfer quantum states from one location to another.  Perhaps an obvious method of doing so is via a series of SWAP gates; however, such a procedure would require a great deal of control over the system and would be highly prone to errors, as analyzed by Petrosyan, Nikolopoulos, and Lambropoulos~\cite{PNL10}.  Instead, one could take advantage of the natural propagation of the system as time passes to transfer quantum states.  The protocol for quantum communication through unmeasured and unmodulated spin chains was presented by Bose~\cite{B03}, and led to the interpretation of quantum channels implemented by spin chains as wires for transmission of states.  We model the spin chains as graphs.

When considering a quantum state of a particle on a graph, one goal we consider is that of \emph{perfect state transfer}, in which the fidelity of the state transfer is equal to one.  This concept of perfect state transfer was introduced by Christandl et al.~\cite{CDEL04, CDDEKL05}, who showed that perfect state transfer on uniformly coupled spin chains is only possible for chains of two or three qubits.  If non-uniform coupling schemes are considered, then perfect state transfer can be realized on spin chains of arbitrary length, as demonstrated by Christandl et al.~\cite{CDDEKL05} and Yung and Bose~\cite{YB05}.  Later, Wang, Shuang, and Rabitz~\cite{WSR11} and Vinet and Zhedanov~\cite{VZ12a} investigated finding all coupling schemes allowing perfect state transfer.  Engineering such a scheme, however, would be highly difficult in practice.  Kempton, Lippner, and Yau~\cite{KLY17} demonstrated that for paths of length at least four, there is no potential (graph theoretically, weighted loops) that permits perfect state transfer to be achieved.

On the other hand, it is worth investigating whether relaxing the requirement on the fidelity would lead to better implementation.  Such a question is motivated by arguments that these implementations are too demanding compared to the level of fidelity required for most tasks in quantum information processing~\cite{Zenchuk2012}.  Hence, \emph{pretty good state transfer} has been introduced by multiple authors, including Vinet and Zhedanov~\cite{VZ12} (under the term almost perfect state transfer) and Godsil~\cite{Godsil2012}, which requires there to be times where the fidelity of transfer is arbitrarily close to one.

Godsil, Kirkland, Servini, and Smith~\cite{GKSS12} considered pretty good state transfer on unweighted paths, and proved the following.  We let $P_n$ denote a path on $n$ vertices, and assume without loss of generality that the vertices of $P_n$ are labelled 1 to $n$ such that vertices with consecutive labels are adjacent.

\begin{theorem} \cite{GKSS12} \label{pgst-path-end}
Pretty good state transfer occurs between the end vertices of $P_n$ if and only if $n = p - 1$, $2p - 1$, where $p$ is a prime, or $n = 2^m - 1$.  Moreover, when pretty good state transfer occurs between the end vertices of $P_n$, then it occurs between vertices $a$ and $n + 1 - a$ for all $a \neq (n + 1) / 2$.
\end{theorem}

Coutinho, Guo, and van Bommel~\cite{CGvB17} determined an infinite family of paths which exhibit pretty good state transfer between internal vertices but not between the end vertices, and van Bommel~\cite{vB19} completed the characterization by showing there were no additional examples, leading to the following result.

\begin{theorem} \label{thm-pgst-path} \cite{vB19}
There is pretty good state transfer on $P_n$ between vertices $a$ and $b$ if and only if $a + b = n + 1$ and either:
\begin{enumerate}[a)]
\item $n = 2^t - 1$, where $t$ is a positive integer;
\item $n = p - 1$, where $p$ is an odd prime; or, 
\item $n = 2^t p - 1$, where $t$ is a positive integer and $p$ is an odd prime, and $a$ is a multiple of $2^{t - 1}$.
\end{enumerate}
\end{theorem}

Casaccino, Llyod, Mancini, and Severini~\cite{CLMS09} observed numerical evidence suggesting pretty good state transfer between the endpoints of a path by adding an appropriate potential at the endpoints.  Kempton, Lippner, and Yau~\cite{KLY17b} verified this observation through the following result.

\begin{theorem} \cite{KLY17b}
Given a path $P_n$, there is some choice of $Q$ such that by placing the value $Q$ as a potential on each endpoint of $P_n$ there is pretty good state transfer between the endpoints.
\end{theorem}

They also provide the following more general result on graphs with an involution (an automorphism of the graph of order two).

\begin{theorem} \cite{KLY17b}
Let $G$ be a connected graph with an involution $\sigma$ and let $Q: V(G) \to \mathbb{R}$ be a potential on the vertex set satisfying $Q(x) = Q(\sigma x)$.  Let $u$ and $v$ be vertices with $v = \sigma u$ and $u \neq v$.  Then if $\sigma$ fixes any vertices or any edges of $G$, then there is a choice of potential $Q$ for which there is pretty good state transfer from $u$ to $v$.
\end{theorem}

Godsil, Guo, Kempton, and Lippner~\cite{GGKLM20} considered the problem of achieving perfect and pretty good state transfer in strongly regular graphs.  A graph $G$ is \emph{strongly regular} if it is neither empty nor complete, it is $k$-regular for some integer $k$, every pair of adjacent vertices has $a$ common neighbours, and every pair of nonadjacent vertices has $c$ common neighbours.  Strongly regular graphs are characterized as the set of regular graphs with exactly three (distinct) eigenvalues; we will denote these eigenvalues by $k$, $\theta$, and $\tau$.  The authors modify a strongly regular graph $G$ by choosing a pair of vertices $x, y$, adding a weighted edge between $x$ and $y$ with weight $\beta$, and adding weighted loops to the vertices $x$ and $y$ with weight $\gamma$, denoting the resulting graph $G^{\beta, \gamma}$ and observe the following.

\begin{theorem} \cite{GGKLM20}
Let $G$ be a strongly regular graph with eigenvalues $k, \theta, \tau$.  Then if $k \equiv \theta \equiv \tau \pmod{4}$ is odd (resp.\ even), then for any pair of nonadjacent (resp.\ adjacent) vertices $x, y$ of $G$ there is a choice of $\beta, \gamma$ such that there is perfect state transfer between $x$ and $y$ in $G^{\beta, \gamma}$.
\end{theorem}

\begin{theorem} \cite{GGKLM20}
Let $G$ be a strongly regular graph.  Then for any pair of vertices $x, y$ of $G$, there is a choice of $\beta, \gamma$ such that there is pretty good state transfer between $x$ and $y$ in $G^{\beta, \gamma}$.
\end{theorem}

Eisenberg, Kempton, and Lippner~\cite{EKL18} generalized these results to asymmetric graphs; we will discuss their main result in the next section, before looking at extensions to this result.

Chen, Mereau, and Feder~\cite{CMF16} considered the problem of improving state transfer on paths by adding weighed connections to new sites from the third site from each end.  They find that using a weight proportional to $\sqrt{n}$, pretty good state transfer is achieved between the end points on the time scale proportional to $n^{3/2}$, with a transfer probability $p$ such that $1 - p$ is proportional to $1 / n$.  We analyze the general problem of adding a weighted connection to each end of a path as an application of the main result of Eisenberg, Kempton, and Lippner~\cite{EKL18} and our extensions on the result.

\section{Preliminaries}

We consider pretty good state transfer on a graph $G$, with vertex set $V(G)$ and edge set $E(G)$, under the model described with an XX-type Hamiltonian in the presence of a magnetic field, given by
\[
H = \frac{1}{2} \sum_{(i, j) \in E(G)} J_{ij} (X_i X_j + Y_i Y_j) + \sum_{i \in V(G)} Q_i Z_i,
\]
where $X_i$, $Y_i$, and $Z_i$ are Pauli matrices acting at position $i$, $J_{ij}$ denotes the coupling strength between $i$ and $j$, and $Q_i$ denotes the magnetic field strength at $i$.  From a graph theoretic standpoint, we can model the coupling strength by weighting the edges of the graph, and model the magnetic field by adding weighted loops at the vertices; this information is captured by the (weighted) adjacency matrix $A$ of the graph.  We consider the entries of $A$ to be contained in some subfield $\mathcal{F}$ of the real numbers.  Using the Jordan-Wigner transformation~\cite{JW}, we obtain that when the initial state is restricted to the single excitation space, the evolution of the process is governed by $U(t) = \exp(i t A)$.

Formally, we say a graph $G$ has \emph{perfect state transfer} from $x$ to $y$ if there is some time $\tau > 0$ such that $|U(\tau)_{x, y}| = 1$.  Further, we say a graph $G$ has \emph{pretty good state transfer} from $x$ to $y$ if, for any $\epsilon > 0$, there exists a time $\tau_\epsilon > 0$ such that $|U(\tau_\epsilon)_{x, y}| > 1 - \epsilon$.

If $M$ is a symmetric matrix with $d$ distinct eigenvalues $\{\theta_j\}_{j = 1}^d$, then the spectral decomposition of $M$ is
\[
M = \sum_{j = 1}^d \theta_j E_j,
\]
where $E_j$ denotes the orthogonal projection onto the eigenspace corresponding to $\theta_j$.  When we speak of the spectral decomposition of a graph $G$, we mean the spectral decomposition of its adjacency matrix $A$.  For a vertex $x$, let $\ee_x$ represent the characteristic vector of $x$.

We say that two vertices $x$ and $y$ of a graph $G$ are \emph{cospectral} if for each idempotent $E_j$ in the spectral decomposition of $G$, we have $(E_j)_{x,x}  = (E_j)_{y, y}$.  Equivalently, we can say that $G \setminus x$ and $G \setminus y$ have the same characteristic polynomial.  We say that two vertices are \emph{parallel} if for each eigenvalue $\theta_j$, the vectors $E_j \ee_x$ and $E_j \ee_y$ are parallel.  We say that two vertices are \emph{strongly cospectral} if for each eigenvalue $\theta_j$, we have $E_j \ee_x = \pm E_j \ee_y$.  A pair of vertices is strongly cospectral if and only if the vertices are cospectral and parallel, as demonstrated by Godsil and Smith~\cite{GS17}.  As discussed in~\cite{Godsil2012}, Dave Witte Morris observed that a necessary condition for two vertices to admit pretty good state transfer is that they be strongly cospectral.

Equivalent to the condition on the projections onto the eigenspace, we can also determine strong cospectrality by considering the eigenvectors, as demonstrated below.

\begin{prop} \label{prop-sc-equiv}
For every eigenvalue $\lambda$ of $H$, $E_\lambda \ee_u = \pm E_\lambda \ee_v$ if and only if every eigenvector $x$ of $H$ with eigenvalue $\lambda$ satisfies $x(u) = \pm x(v)$.
\end{prop}

\begin{proof}
Suppose $E_\lambda \ee_u = \pm E_\lambda \ee_v$ and let $x$ be an eigenvector of $H$ with eigenvalue $\lambda$.  Then 
\[
x(u) = \ee_u^\dagger x = \ee_u^\dagger E_\lambda x = (E_\lambda \ee_u)^\dagger x = \pm (E_\lambda \ee_v)^\dagger x = \pm \ee_v^\dagger E_\lambda x = \pm \ee_v^\dagger x = \pm x(v)
\]
as desired.  Conversely, if every eigenvector $x$ of $H$ with eigenvalue $\lambda$ satisfies $x(u) = \pm x(v)$, then we can calculate
\[
E_\lambda = \sum_i x_i x_i^T,
\]
where $\{x_i\}$ is an orthonormal basis of eigenvectors with eigenvalue $\lambda$, and then
\[
E_\lambda \ee_u = \sum_i x_i x_i^T \ee_u = \sum_i x_i x_i(u) = \sum_i x_i (\pm x_i(v)) = \pm \sum_i x_i x_i^T \ee_v = \pm E_\lambda \ee_v
\]
as desired.
\end{proof}

Let $P_+$ be the minimal polynomial of $A$ relative to $\ee_x + \ee_y$ and let $P_-$ be the minimal polynomial of $A$ relative to $\ee_x - \ee_y$, where the \emph{minimal polynomial of $A$ relative to $z$} is the smallest degree monic polynomial $\rho$ such that $\rho(A) z = 0$. Let $\tr(\rho)$ denotes the trace (i.e.\ the sum of roots) of a polynomial $\rho$.  Eisenberg, Kempton, and Lippner~\cite{EKL18} provided the following factorization of the characteristic polynomial of a graph with cospectral vertices using these polynomials.

\begin{lemma} \cite{EKL18} \label{charpolyfactor}
Let $G$ be a (weighted) graph and let $x, y$ be cospectral vertices of $G$.  The characteristic polynomial $\phi$ of $G$ decomposes as 
\[
\phi = P_+ \cdot P_- \cdot P_0,
\]
where $P_+$ and $P_-$ have no multiple roots, and there is an orthonormal basis of eigenvectors of $G$ such that:
\begin{enumerate}
\item for each root $\lambda$ of $P_+$ the basis contains a unique eigenvector $\varphi$ with eigenvalue $\lambda$ and $\varphi(x) = \varphi(y) \neq 0$,
\item for each root $\lambda$ of $P_-$ the basis contains a unique eigenvector $\varphi$ with eigenvalue $\lambda$ and $\varphi(x) = - \varphi(y) \neq 0$,
\item for each root $\lambda$ of $P_0$ with multiplicity $k$ the basis contains exactly $k$ eigenvectors with eigenvalue $\lambda$ all of which vanish on both $x$ and $y$.
\end{enumerate}
\end{lemma}

When the vertices $x$ and $y$ are strongly cospectral, we observe that $P_+$ and $P_-$ have no common roots.

The following lemma provides a characterization of pretty good state transfer.  Its form is due to Kempton, Lippner, and Yau~\cite{KLY17b}, which was derived from results of Banchi, Coutinho, Godsil, and Severini~\cite{BCGS16}.

\begin{lemma}\cite{KLY17b, BCGS16} \label{lem:eig}
Let $G$ be a graph and let $x, y$ be vertices of $G$.  Then pretty good state transfer from $x$ to $y$ occurs if and only if the following two conditions are satisfied:
\begin{enumerate}
\item The vertices $x$ and $y$ are strongly cospectral.
\item Let $\{\lambda_i\}$ be the roots of $P_+$ and $\{\mu_j\}$ the roots of $P_-$.  Then for any choice of integers $\ell_i, m_j$ such that
\begin{align*}
\sum_i \ell_i \lambda_i + \sum_j m_j \mu_j &= 0 \\
\sum_i \ell_i + \sum_j m_j &= 0,
\end{align*}
we have $\sum_j m_j$ is even.
\end{enumerate}
\end{lemma}

Eisenberg, Kempton, and Lippner~\cite{EKL18} obtain the following sufficient condition for pretty good state transfer by considering the minimal polynomials $P_+$ and $P_-$.

\begin{theorem} \label{pgst-asym}
Let $G$ be a graph, let $x$ and $y$ be strongly cospectral vertices of $G$, and assume that $P_+$ and $P_-$ are irreducible polynomials over $\mathcal{F}$.  Then if
\[
\frac{\tr(P_+)}{\deg(P_+)} \neq \frac{\tr(P_-)}{\deg(P_-)},
\]
then there is pretty good state transfer from $x$ to $y$.
\end{theorem}

The results of Kempton, Lippner, and Yau~\cite{KLY17b} and Eisenberg, Kempton, and Lippner~\cite{EKL18} focus on determining choices of potentials to ensure general classes of graphs can admit pretty good state transfer.  On the other hand, we observe instances of unweighted graphs that satisfy the conditions of Theorem~\ref{pgst-asym}, and provide examples in Appendix~\ref{app:unweighted}.  Unfortunately, we know of no general method of producing graphs with irreducible minimal polynomials.

\section{Involutions}

We have seen in Lemma~\ref{lem:eig} that pretty good state transfer in a graph is characterized by the strong cospectrality of the vertices and a restriction on the (rational) dependence of the eigenvalues in the support.  Kempton, Lippner, and Yau~\cite{KLY17b} claimed that a graph with an involution always has the strong cospectrality condition.  We will show that this is not entirely correct, but note that we can use the characteristic polynomial factorization provided by Lemma~\ref{charpolyfactor} to determine strong cospectrality.

Let $G$ be a graph with weighted adjacency matrix $A$ and non-trivial involution $\sigma$, such that the weights (including loops) are symmetric with respect to the involution.  Let $S = \{x \in V(G) : \sigma x = x\}$ be the set of vertices fixed by $\sigma$ and let $G'$ be induced by a set of vertices of $G \setminus S$ such that precisely one of $x$ and $\sigma x$ is in $V(G')$ for all $x \in V(G \setminus S)$.  Observe that $G - G' - S \equiv G'$.  We can express the (weighted) adjacency matrix $A$ of $G$ as
\[
A = \begin{bmatrix} A' & A_\sigma & A_\delta \\ A_\sigma & A' & A_\delta \\ A_\delta^T & A_\delta^T & A_S \end{bmatrix},
\]
where $A'$ denotes the (weighted) adjacency matrix of $G'$, $A_\sigma$ the portion of $A$ representing edges between $G'$ and $G - G' - S$, $A_\delta$ the portion of $A$ representing edges between $G'$ and $S$, and $A_S$ the (weighted) adjacency matrix of $G[S]$.

We present the following lemma due to Kempton, Lippner, and Yau~\cite{KLY17b} which describes how the characteristic polynomial factors.

\begin{lemma} \cite{KLY17b} \label{invol_factor}
Let $G$ be a (weighted) graph with an involution $\sigma$, which respects loops and edge weights.  Then the characteristic polynomial of the (weighted) adjacency matrix $A$ of $G$ factors into two factors $\Pi_+$ and $\Pi_-$ which are, repsectively, the characteristic polynomials of $A_+ := \begin{bmatrix} A' + A_\sigma & A_\delta \\ 2 A_\delta^T & A_S \end{bmatrix}$ and $A_- := A' - A_\sigma$.

Furthermore, there is an eigenbasis for $A$ consisting of vectors that take the form $\begin{bmatrix} a & a & b \end{bmatrix}^T$ and $\begin{bmatrix} c & -c & 0 \end{bmatrix}^T$, where $\begin{bmatrix} a & b \end{bmatrix}^T$ is an eigenvector for $A_+$, and $c$ an eigenvector for $A_-$.
\end{lemma}

We note that $P_+$ is a factor of $\Pi_+$ and $P_-$ is a factor of $\Pi_-$.  By the following example, we demonstrate that the presence of an involution is not sufficient to guarantee strong cospectrality.

\begin{example}
Let $G = K_3$ and $\sigma$ be the involution swapping two of the vertices and leaving the third fixed.  Then we have
\[
A = \begin{bmatrix} 0 & 1 & 1 \\ 1 & 0 & 1 \\ 1 & 1 & 0 \end{bmatrix}, \qquad A_+ = \begin{bmatrix} 1 & 1 \\ 2 & 0 \end{bmatrix}, \quad A_- = \begin{bmatrix} -1 \end{bmatrix},
\]
so $\Pi_+(x) = x^2 - x + 2$ and $\Pi_-(x) = x + 1$.  Hence,
\[ \left\{ \begin{bmatrix} 1 & 1 & 1 \end{bmatrix}^T, \begin{bmatrix}1 & 1 & -2\end{bmatrix}^T, \begin{bmatrix} 1 & -1 & 0 \end{bmatrix}^T \right\}, \]
is an eigenbasis for $A$ where $\left\{ \begin{bmatrix} 1 & 1 \end{bmatrix}^T, \begin{bmatrix} 1 & -2 \end{bmatrix}^T \right\}$ is an eigenbasis for $A_+$ and $\left\{\begin{bmatrix} 1 \end{bmatrix}\right\}$ is an eigenbasis for $A_-$.  However, the two vertices swapped by the involution are not strongly cospectral; we observe that $\begin{bmatrix} 2 & 0 & -2 \end{bmatrix}^T$ is also an eigenvector of $A$ but does not satisfy the strong cospectrality condition. \hfill $\Diamond$
\end{example}

The vertices swapped by the involution in the above example fail to be strongly cospectral because $\Pi_+$ and $\Pi_-$ have a common factor leading to a repeated eigenvalue.  In the next example, we show that $\Pi_+$ and $\Pi_-$ depend on the choice of involution, but $P_+$ and $P_-$ do not.

\begin{example}
Let $G = K_{2, 2}$ and $\sigma$ be the involution swapping the vertices in each part.  Then we have
\[
A = \begin{bmatrix} 0 & 1 & 0 & 1 \\ 1 & 0 & 1 & 0 \\ 0 & 1 & 0 & 1 \\ 1 & 0 & 1 & 0 \end{bmatrix}, \qquad A_+ = \begin{bmatrix} 0 & 2 \\ 2 & 0 \end{bmatrix}, \quad A_- = \begin{bmatrix} 0 & 0 \\ 0 & 0 \end{bmatrix},
\]
so $\Pi_+(x) = x^2 - 4$ and $\Pi_-(x) = x^2$.  On the other hand, letting $\sigma'$ be the involution that swapping the vertices in one part and leaving the other fixed, we have
\[
A = \begin{bmatrix} 0 & 0 & 1 & 1 \\ 0 & 0 & 1 & 1 \\ 1 & 1 & 0 & 0 \\ 1 & 1 & 0 & 0 \end{bmatrix}, \qquad A_+ = \begin{bmatrix} 0 & 1 & 1 \\ 2 & 0 & 0 \\ 2 & 0 & 0 \end{bmatrix}, \quad A_- = \begin{bmatrix} 0 \end{bmatrix},
\]
so $\Pi_+(x) = x^3 - 4x$ and $\Pi_-(x) = x$.  

On the other hand, the spectral decomposition of $A$ is given by
\begingroup
\renewcommand*{\arraystretch}{1.5}
\[
A = \begin{bmatrix} 0 & 0 & 1 & 1 \\ 0 & 0 & 1 & 1 \\ 1 & 1 & 0 & 0 \\ 1 & 1 & 0 & 0 \end{bmatrix} = 2 \begin{bmatrix} \frac{1}{4} & \frac{1}{4} & \frac{1}{4} & \frac{1}{4} \\ \frac{1}{4} & \frac{1}{4} & \frac{1}{4} & \frac{1}{4} \\ \frac{1}{4} & \frac{1}{4} & \frac{1}{4} & \frac{1}{4} \\ \frac{1}{4} & \frac{1}{4} & \frac{1}{4} & \frac{1}{4} \end{bmatrix}  - 2 \begin{bmatrix} \frac{1}{4} & \frac{1}{4} & -\frac{1}{4} & -\frac{1}{4} \\ \frac{1}{4} & \frac{1}{4} & -\frac{1}{4} & -\frac{1}{4} \\ -\frac{1}{4} & -\frac{1}{4} & \frac{1}{4} & \frac{1}{4} \\ -\frac{1}{4} & -\frac{1}{4} & \frac{1}{4} & \frac{1}{4} \end{bmatrix} + 0 \begin{bmatrix} \frac{1}{2} & - \frac{1}{2} & 0 & 0 \\ - \frac{1}{2} & \frac{1}{2} & 0 & 0 \\ 0 & 0 & \frac{1}{2} & - \frac{1}{2} \\ 0 & 0 & - \frac{1}{2} & \frac{1}{2} \end{bmatrix},
\]
\endgroup

so we obtain $P_+(x) = x^2 - 4$ and $P_-(x) = x$. \hfill $\Diamond$
\end{example}

Making use of the involution directly, Kempton, Lippner, and Yau~\cite{KLY17b} provide a condition on the characteristic polynomial factors that allows pretty good state transfer to be achieved with potentials on only those two vertices.

\begin{cor} \cite{KLY17b} \label{loop-factor}
Let $G$ be a graph with an involution $\sigma$, a vertex $x$ that is not fixed by $\sigma$, and suppose the only vertices with a potential $Q$ are $x$ and $\sigma x$, with $Q$ transcendental.  Let
\[
\Pi_+(x) = p_1(x) - Q q_1(x), \quad \Pi_-(x) = p_2(x) - Q q_2(x),
\]
where $p_1, q_1, p_2, q_2$ are rational polynomials.  Suppose that $\gcd(p_1, q_1) = \gcd(p_2, q_2) = 1$.  If $\sigma$ fixes at least one vertex or at least one edge, then there is pretty good state transfer from $x$ to $\sigma x$.
\end{cor}

Moreover, Kempton, Lippner, and Yau~\cite{KLY17b} provide a partial general result for determining pretty good state transfer when there are no fixed vertices or edges.

\begin{prop} \cite{KLY17b} \label{prop:evenodd}
Let $G$ be a graph with an involution $\sigma$ that has no fixed vertices or edges, and let the number of vertices of $G$ be $N = 2n$.  Let $Q$ be a potential on $G$ that is symmetric across $\sigma$.  Then we have the following.
\begin{enumerate}
\item If $n$ is even, and if $\Pi_+$ is irreducible and the splitting fields for $\Pi_+$ and $\Pi_-$ intersect only in the base field, then there is pretty good state transfer between $x$ and $\sigma x$ for all $x$.
\item If $n$ is odd, and if all the eigenvectors of $A$ are non-vanishing on vertices $x, \sigma x$, then there is not pretty good state transfer between $x$ and $\sigma x$.
\end{enumerate}
\end{prop}

We will make use of these results in the sections that follow.

\section{Extensions}

We consider extensions to Theorem~\ref{pgst-asym} to resolve more cases of pretty good state transfer based on the minimal polynomials.  We first note that the condition on the irreducibility of polynomials $P_+$ and $P_-$ cannot, in general, be removed, as demonstrated below.

\begin{example}
By Theorem~\ref{pgst-path-end}, $P_8$ does not have pretty good state transfer between its end vertices $1$ and $8$.  We have that $P_+ = (x - 1)(x^3 - 3x - 1)$ and $P_- = (x + 1) (x^3 - 3x + 1)$.  Hence the polynomials are reducible but all other conditions of the theorem are satisfied. \hfill $\Diamond$
\end{example}

We note that this example generalizes to $P_n$ where $n + 1$ is an odd composite number.  We now consider the following result which extends the characterization of pretty good state transfer when $P_+$ and $P_-$ are irreducible, and additionally provides a natural extension to Proposition~\ref{prop:evenodd}(2).

\begin{theorem} \label{thm:odd-no-pgst}
Let $G$ be a graph, let $x$ and $y$ be strongly cospectral vertices of $G$, and suppose $P_+$ and $P_-$ have (possibily trivial) factors $f_+$ and $f_-$ of odd degree.  Then if
\[
\frac{\tr(f_+)}{\deg(f_+)} = \frac{\tr(f_-)}{\deg(f_-)},
\]
then there is no pretty good state transfer from $x$ to $y$.
\end{theorem}

\begin{proof}
Let 
\[
\ell_i = \begin{cases}
- \deg(f_-) , & f_+(\lambda_i) = 0; \\
0, & f_+(\lambda_i) \neq 0;
\end{cases}
\qquad
m_j = \begin{cases}
\deg(f_+), & f_-(\mu_j) = 0; \\
0, & f_-(\mu_j) \neq 0.
\end{cases}
\]
Then we have
\begin{align*}
\sum_i \ell_i \lambda_i + \sum_j m_j \mu_j &= - \deg(f_-) \tr(f_+) + \deg(f_+) \tr(f_-) = 0, \\
\sum_i \ell_i + \sum_j m_j &= - \deg(f_+) \deg(f_-) + \deg(f_+) \deg(f_-) = 0,
\end{align*}
and since $\sum_j m_j$ is odd, then it follows from Lemma~\ref{lem:eig} that we do not have pretty good state transfer between $x$ and $y$.
\end{proof}

\begin{example}
Let $W = C_4 \vee K_1$, and consider a pair of non-adjacent vertices.  The minimal polynomials are $P_+(x) = (x + 2) (x^2 - 2x - 4) = x^3 - 8x - 8$ and $P_-(x) = x$.  By Theorem~\ref{thm:odd-no-pgst}, $W$ does not have pretty good state transfer between these vertices.
\end{example}

For our next extension, the polynomials are not irreducible, and we cannot find factors to satisfy the result above, but with mild additional conditions we can again rule out pretty good state transfer.

\begin{theorem} \label{no-3-factors}
Let $G$ be a graph and $x, y$ a pair of strongly cospectral vertices of $G$.  If one of $P_+, P_-$ is divisible by the product of nontrivial polynomials $f$ and $g$ such that $f$ has odd trace and $\tr(f) \deg(g) - \tr(g) \deg(f)$ is odd, and the other has a (possibly trivial) factor $h$ of odd degree, then $G$ does not have pretty good state transfer between $x$ and $y$.
\end{theorem}

\begin{proof}
Without loss of generality, we may assume $P_+$ is divisible by $f * g$ and $P_-$ has such a factor $h$.  Let $\tr(f) \deg(g) - \tr(g) \deg(f) = 2p + 1$, $\tr(f) \deg(h) - \tr(h) \deg(f) = q$, and
\[
\ell_i = \begin{cases}
(2p + 1) \tr(h) - q \tr(g), & f(\lambda_i) = 0; \\
q \tr(f), & g(\lambda_i) = 0; \\
0, & \mbox{otherwise;}
\end{cases} \qquad
m_j = \begin{cases}
- (2p + 1) \tr(f), & h(\mu_j) = 0; \\
0, & h(\mu_j) \neq 0.
\end{cases}
\]
Then we have
\begin{align*}
&\qquad \sum_i \ell_i \lambda_i + \sum_j m_j \mu_j \\ &= \sum_{i: f(\lambda_i) = 0} ((2p + 1) \tr(h) - q \tr(g)) \lambda_i + \sum_{i: g(\lambda_i) = 0} q \tr(f) \lambda_i + \sum_{j: h(\mu_j) = 0} - (2p + 1) \tr(f) \mu_j \\
&= (2p + 1) \tr(h) \tr(f) - q \tr(g) \tr(f) + q \tr(f) \tr(g) - (2p + 1) \tr(f) \tr(h) = 0; \\
&\qquad \sum_i \ell_i + \sum_j m_j  \\ &= \sum_{i: f(\lambda_i) = 0} (2p + 1) \tr(h) - q \tr(g) + \sum_{i: g(\lambda_i) = 0} q \tr(f) + \sum_{j: h(\mu_j) = 0} - (2p + 1) \tr(f) \\
&= (2p + 1) \tr(h) \deg(f) - q \tr(g) \deg(f) + q \tr(f) \deg(g) - (2p + 1) \tr(f) \deg(h) = 0 \\
\end{align*}
and since $\sum_j m_j$ is odd by construction, it follows from Lemma~\ref{lem:eig} that we do not have pretty good state transfer between $x$ and $y$.
\end{proof}

\begin{example}
Let $H$ be the tree with two adjacent vertices of degree 3, and the remaining vertices of degree 1, and consider two leaves with a common neighbour.  The minimal polynomials are $P_+(x) = (x - 2) (x - 1) (x + 1) (x + 2)$ and $P_-(x) = x$.  We let $h = x$, $f = x - 1$, and $g = (x + 1)(x + 2) = x^2 + 3x + 2$.  Then by Theorem~\ref{no-3-factors}, $H$ does not have pretty good state transfer between these vertices.
\end{example}

\section{Modified Paths}

Let $P_N^{(M, w)}$ denote the path of length $N$, with vertices labeled 1 to $N$, and with an additional vertex labeled $0$ joined to vertex $M$ and an additional vertex labeled $N + 1$ joined to vertex $N + 1 - M$, each by edges of weight $w$.  We consider pretty good state transfer on this modified path.  We first present the following variant of Corollary~\ref{loop-factor} to consider a weighted edge rather than a loop.

\begin{lemma}
Let $G$ be a graph with an involution $\sigma$ and suppose precisely one edge $e$ of $G'$ has transcendental weight $w$, with all other weights of $G$ rational.  Let $\Pi_\pm(x) = g_\pm(x) (p_\pm(x) - w^2 q_\pm(x))$, where $g_\pm(x) p_\pm(x)$ is the characteristic polynomial of $A_\pm$ with $w = 0$, $g_\pm(x) q_\pm(x)$ is the characteristic polynomial of $A_\pm$ with the rows and columns indexed by the endpoints of $e$ deleted, and $\gcd(p_\pm(x), q_\pm(x)) = 1$.  Then every irreducible factor of $\Pi_\pm(x)$ is a factor of $g_\pm(x)$, or is either $p_\pm(x) - w^2 q_\pm(x)$ if one of $p_\pm(x), q_\pm(x)$ are not perfect squares, or one of $\sqrt{p_\pm(x)} - w \sqrt{q_\pm(x)}$, $\sqrt{p_\pm(x)} + w \sqrt{q_\pm(x)}$.
\end{lemma} 

\begin{proof}
Suppose $\Pi_\pm(x)$ has an irreducible factor $f(x)$ that is not a factor of $g_\pm(x)$.  If $f(x) \in \mathbb{Q}[x]$, then $f(x)$ is a factor of both $p_\pm(x)$ and $q_\pm(x)$, contradicting that $\gcd(p_\pm(x), q_\pm(x)) = 1$.  
Now, if $f(x)$ is linear in $w$, then we must have $\Pi_\pm(x) = g(x) (a(x) + w b(x))(c(x) - w d(x))$.  It is clear that $\gcd(a(x), b(x)) = \gcd(c(x), d(x)) = 1$, as otherwise $\gcd(p_\pm(x), q_\pm(x)) \neq 1$.  Moreover, we have $a(x) d(x) = b(x) c(x)$, from which it follows that $a(x) = c(x) = \sqrt{p_\pm(x)}$ and $b(x) = d(x) = \sqrt{q_\pm(x)}$.  Finally, if $f(x)$ is quadratic in $w$, then it must be precisely $p_\pm(x) - w^2 q_\pm(x)$, as we have shown this expression has no other factors.
\end{proof}

We are now able to determine the existence of pretty good state transfer in the following cases.

\begin{theorem}
Let $w$ be transcendental.  Then 
\begin{enumerate}
\item $P_N^{(M, w)}$ has pretty good state transfer between vertices 1 and $N$ if $M$ is odd, $N$ is even, and $\gcd(N + 1, M) = 1$.  
\item $P_N^{(M, w)}$ does not have pretty good state transfer between vertices 1 and $N$ if:
\begin{enumerate}
\item $M$ is even;
\item $M$ is odd, $N$ is even but not divisible by 4, and $\gcd(N + 1, M)$ has a prime factor of the form $4r + 3$.
\end{enumerate}
\end{enumerate}
\end{theorem}

\begin{proof}
Since $P_N^{(M, w)}$ has an involution which takes vertex $x$ to vertex $N + 1 - x$, then the characteristic polynomial of $P_N^{(M, w)}$ can be factored into $\Pi_+$ and $\Pi_-$, where the eigenvectors corresponding to eigenvalues of $\Pi_+$ are symmetric and the eigenvectors corresponding to eigenvalues of $\Pi_-$ are antisymmetric.  Let $p_N(x)$ be the characteristic polynomial of a (unmodified) path on $N$ vertices.  If $N = 2n$ is even, then we obtain
\begin{align*}
\Pi_+(x) &= x (p_n(x) - p_{n - 1}(x)) - w^2 p_{M - 1}(x) (p_{n - M}(x) - p_{n - M - 1}(x)),  \\
\Pi_-(x) &= x (p_n(x) + p_{n - 1}(x)) - w^2 p_{M - 1}(x) (p_{n - M}(x) + p_{n - M - 1}(x));
\end{align*}
and if $N = 2n + 1$ is odd, then we obtain
\begin{align*}
\Pi_+(x) &= x (p_{n + 1}(x) - p_{n - 1}(x)) - w^2 p_{M - 1}(x) (p_{n - M + 1}(x) - p_{n - M - 1}(x)), \\
\Pi_-(x) &= x p_n(x) - w^2 p_{M - 1}(x) p_{n - M}(x).
\end{align*}

Suppose $M$ is even.  We show that the vertices 1 and $N$ are not strongly cospectral.  Consider the vector $z$ defined for $0 \le j \le M$ by
\[
z(j) = \begin{cases}
1, & j \equiv 1 \pmod{4}, j < M; \\
-1, & j \equiv 3 \pmod{4}, j < M; \\
\frac{(-1)^{\frac{M}{2}}}{w}, & j = 0; \\
0, & \mbox{otherwise},
\end{cases}
\]
and let $y$ be the vector $z$ with the entries in reverse order.  Then $w_+ = \begin{bmatrix} z & 0 & y \end{bmatrix}^T$ and $w_- = \begin{bmatrix} z & 0 & -y \end{bmatrix}$ are both eigenvectors with eigenvalue 0, but $w_+(1) = w_+(N)$ and $w_-(1) = - w_-(N)$.  Hence it follows by Proposition~\ref{prop-sc-equiv} that 1 and $N$ are not strongly cospectral, and hence cannot admit pretty good state transfer.

We may now assume that $M$ is odd.  We wish to determine $P_+$ and $P_-$.  Suppose $\lambda$ is an eigenvalue of $P_N^{(M, w)}$ that is not in the eigenvalue support of vertex 1.  Then for every eigenvector $z$ with eigenvalue $\lambda$, we have that $z(1) = 0$.  It follows that $z(j) = 0$ for all $1 \le j \le M$ and all $N - M < j \le N$.  If $z(0) = 0$, then $z = 0$, contradicting that $z$ is an eigenvector.  It follows that $\lambda = 0$, $N$ is odd, and every odd vertex $j$ ($1 \le j \le N$) has $z(j) = 0$.  At this point we conclude that $\Pi_+ = P_+$ and $\Pi_- = P_-$ for $N$ even.  For $N$ odd, we have two cases.  If $n$ is even, then $z$ is antisymmetric, and $\Pi_+ = P_+$ and $\Pi_- = x P_-$.  If $n$ is odd, then $z$ is symmetric, and $\Pi_+ = x P_+$ and $\Pi_- = P_-$.

We consider the irreducibility of these polynomials.  We begin with the case when $N$ is even.  We observe that $p_k(x) \pm p_{k - 1}(x)$ are factors of $p_{2k}(x)$ and that 0 is a root of $p_k(x)$ if and only if $k$ is odd.  Hence, we have 
\[
\gcd(x (p_n(x) \pm p_{n - 1}(x)), p_{M - 1}(x) (p_{n - M}(x) \pm p_{n - M - 1}(x))) \mid \gcd(p_{2n}(x), p_{M - 1}(x) p_{2n - 2M}(x)).
\]
Now, $p_k(x)$ and $p_\ell(x)$ have common roots if and only if $\gcd(k + 1, \ell + 1) > 1$.  So if $p_{2n}(x)$ has a common root with $p_{M - 1}(x)$, then $\gcd(2n + 1, M) > 1$, and if $p_{2n}(x)$ has a common root with $p_{2n - 2M}(x)$, then $\gcd(2n + 1, 2n - 2M + 1) = \gcd(2n + 1, 2M) = \gcd(2n + 1, M) > 1$.  Hence, if $\gcd(2n + 1, M) = 1$, then $P_\pm$ is irreducible.  Moreover, we have that $\deg(P_+) = \deg(P_-)$ and $\tr(P_+) \neq \tr(P_-)$, so by Theorem~\ref{pgst-asym}, there is pretty good state transfer between vertices 1 and $N$.

Now, suppose $\gcd(2n + 1, M)$ is divisible by a prime $q$ such that $q = 4r + 3$.  Then, $p_{4r + 2}(x) \mid p_{M - 1}(x)$, and it follows that $p_{2r + 1}(x) - p_{2r}(x) \mid P_+$ and $p_{2r + 1}(x) + p_{2r}(x) \mid P_-$.  Then if $n$ is odd, it follows by Theorem~\ref{no-3-factors} that we do not obtain pretty good state transfer between vertices 1 and $N$.
\end{proof}

\appendix

\section{Applying Theorem~\ref{pgst-asym} to Loopless Unweighted Graphs}
\label{app:unweighted}

\begin{tabular}{|c|c|c|c|}
\hline \\

\makecell{\begin{tikzpicture}
\node(A){};
\node[below=of A](B){};
\path (A) edge (B);
\end{tikzpicture}
\\
{$\!\begin{aligned}
P_+(x) &= x - 1 \\
P_-(x) &= x + 1
\end{aligned}$}} &

\makecell{\begin{tikzpicture}
\node(A){};
\node[right=of A](B){};
\node[below=of A](C){};
\node[below=of B](D){};
\path (A) edge (B);
\path (C) edge (D);
\path (A) edge (C);
\end{tikzpicture}
\\
{$\!\begin{aligned}
P_+(x) &= x^2 - x - 1 \\
P_-(x) &= x^2 + x - 1 \\
\end{aligned}$}} &

\makecell{\begin{tikzpicture}
\node(A){};
\node[right=of A](B){};
\node[right=of B](C){};
\node[below=of A](D){};
\node[below=of B](E){};
\node[below=of C](F){};
\path (A) edge (B);
\path (B) edge (C);
\path (D) edge (E);
\path (E) edge (F);
\path (A) edge (D);
\end{tikzpicture}
\\
{$\!\begin{aligned}
P_+(x) &= x^3 - x^2 - 2x + 1 \\
P_-(x) &= x^3 + x^2 - 2x - 1 \\
\end{aligned}$}} &

\makecell{\begin{tikzpicture}
\node(A){};
\node[right=of A](B){};
\node[right=of B](C){};
\node[below=of A](D){};
\node[below=of B](E){};
\node[below=of C](F){};
\path (A) edge (B);
\path (B) edge (C);
\path (D) edge (E);
\path (E) edge (F);
\path (A) edge (D);
\path (B) edge (E);
\end{tikzpicture}
\\
{$\!\begin{aligned}
P_+(x) &= x^3 - 2x^2 - x + 1 \\
P_-(x) &= x^3 + 2x^2 - x - 1 \\
\end{aligned}$}}

\\
\hline
\end{tabular}

\begin{tabular}{|c|c|c|}
\hline \\
\makecell{\begin{tikzpicture}
\node(A){};
\node[right=of A](B){};
\node[right=of B](C){};
\node[right=of C](D){};
\node[below=of A](E){};
\node[below=of B](F){};
\node[below=of C](G){};
\node[below=of D](H){};
\path (A) edge (B);
\path (B) edge (C);
\path (C) edge (D);
\path (E) edge (F);
\path (F) edge (G);
\path (G) edge (H);
\path (B) edge (F);
\end{tikzpicture}
\\
{$\!\begin{aligned}
P_+(x) &= x^4 - x^3 - 3 x^2 + x + 1 \\
P_-(x) &= x^4 + x^3 - 3 x^2 - x + 1
\end{aligned}$}} & 

\makecell{\begin{tikzpicture}
\node(A){};
\node[right=of A](B){};
\node[right=of B](C){};
\node[right=of C](D){};
\node[below=of A](E){};
\node[below=of B](F){};
\node[below=of C](G){};
\node[below=of D](H){};
\path (A) edge (B);
\path (B) edge (C);
\path (C) edge (D);
\path (E) edge (F);
\path (F) edge (G);
\path (G) edge (H);
\path (A) edge (E);
\path (C) edge (G);
\end{tikzpicture}
\\
{$\!\begin{aligned}
P_+(x) &= x^4 - 2 x^3 - 2 x^2 + 3 x + 1 \\
P_-(x) &= x^4 + 2 x^3 - 2 x^2 - 3 x + 1
\end{aligned}$}} & 

\makecell{\begin{tikzpicture}
\node(A){};
\node[right=of A](B){};
\node[right=of B](C){};
\node[right=of C](D){};
\node[below=of A](E){};
\node[below=of B](F){};
\node[below=of C](G){};
\node[below=of D](H){};
\path (A) edge (B);
\path (B) edge (C);
\path (C) edge (D);
\path (E) edge (F);
\path (F) edge (G);
\path (G) edge (H);
\path (A) edge (E);
\path (B) edge (F);
\path (D) edge (H);
\end{tikzpicture}
\\
{$\!\begin{aligned}
P_+(x) &= x^4 - 2 x^3 - 2 x^2 + 3 x + 1 \\
P_-(x) &= x^4 + 2 x^3 - 2 x^2 - 3 x + 1
\end{aligned}$}}
\\ \hline
\end{tabular}

\begin{tabular}{|c|c|c|}
\hline \\
\makecell{\begin{tikzpicture}
\node(A){};
\node[right=of A](B){};
\node[right=of B](C){};
\node[right=of C](D){};
\node[below=of A](E){};
\node[below=of B](F){};
\node[below=of C](G){};
\node[below=of D](H){};
\path (A) edge[bend left] (C);
\path (A) edge (B);
\path (B) edge (C);
\path (C) edge (D);
\path (E) edge[bend right] (G);
\path (E) edge (F);
\path (F) edge (G);
\path (G) edge (H);
\path (A) edge (E);
\end{tikzpicture}
\\
{$\!\begin{aligned}
P_+(x) &= x^4 - x^3 - 4 x^2 + 1 \\
P_-(x) &= x^4 + x^3 - 4 x^2 - 4 x + 1
\end{aligned}$}} &

\makecell{\begin{tikzpicture}
\node(A){};
\node[right=of A](B){};
\node[right=of B](C){};
\node[right=of C](D){};
\node[below=of A](E){};
\node[below=of B](F){};
\node[below=of C](G){};
\node[below=of D](H){};
\path (A) edge[bend left] (C);
\path (A) edge (B);
\path (B) edge (C);
\path (C) edge (D);
\path (E) edge[bend right] (G);
\path (E) edge (F);
\path (F) edge (G);
\path (G) edge (H);
\path (A) edge (E);
\path (C) edge (G);
\path (D) edge (H);
\end{tikzpicture}
\\
{$\!\begin{aligned}
P_+(x) &= x^4 - 3 x^3 - x^2 + 3 x + 1 \\
P_-(x) &= x^4 + 3 x^3 - x^2 - 7 x - 3
\end{aligned}$}} &

\makecell{\begin{tikzpicture}
\node(A){};
\node[right=of A](B){};
\node[right=of B](C){};
\node[right=of C](D){};
\node[below=of A](E){};
\node[below=of B](F){};
\node[below=of C](G){};
\node[below=of D](H){};
\path (A) edge[bend left] (C);
\path (B) edge[bend right] (D);
\path (A) edge (B);
\path (B) edge (C);
\path (C) edge (D);
\path (E) edge[bend right] (G);
\path (F) edge[bend left] (H);
\path (E) edge (F);
\path (F) edge (G);
\path (G) edge (H);
\path (A) edge (E);
\path (B) edge (F);
\end{tikzpicture}
\\
{$\!\begin{aligned}
P_+(x) &= x^4 - 2 x^3 - 4 x^2 + x + 1 \\
P_-(x) &= x^4 + 2 x^3 - 4 x^2 - 9 x - 3
\end{aligned}$}}

\\ \hline
\end{tabular}

\bibliographystyle{plain}

\end{document}